\documentclass[a4paper,10pt]{article}
\usepackage[utf8x]{inputenc}
\usepackage{amssymb}
\usepackage{amsthm}
\usepackage{pstricks,pst-node,pst-plot}
\usepackage{enumerate}
\usepackage{booktabs}
\usepackage{todonotes}
\reversemarginpar

\title{A Note on Connected Dominating Set in Graphs Without Long Paths And Cycles}
\author{Eglantine Camby\\
Universit\'e Libre de Bruxelles\\
D\'epartement de Math\'ematique\\
Boulevard du Triomphe, 1050 Brussels, Belgium\\
\texttt{ecamby@ulb.ac.be}\\~\\
Oliver Schaudt\thanks{Parts of this research have been carried out during the visit of Oliver Schaudt to Universit\'e Libre de Bruxelles, and later the visit of Eglantine Camby to Universit\'e Pierre et Marie Curie.}\\
Universit\'e Pierre et Marie Curie\\
Combinatoire et Optimisation\\
4 place Jussieu, 75252 Paris, France\\
\texttt{schaudt@math.jussieu.fr}}

\begin{document}

\newtheorem{theorem}{Theorem}
\newtheorem{lemma}{Lemma}
\newtheorem{observation}{Observation}
\newtheorem{corollary}{Corollary}
\newtheorem{claim}{Claim}
\newtheorem{conjecture}{Conjecture}

\maketitle

\begin{abstract}

The ratio of the connected domination number, $\gamma_c$, and the domination number, $\gamma$, is strictly bounded from above by 3.
It was shown by Zverovich that for every connected $(P_5,C_5)$-free graph, $\gamma_c = \gamma$.

In this paper, we investigate the interdependence of $\gamma$ and $\gamma_c$ in the class of $(P_k,C_k)$-free graphs, for $k \ge 6$.
We prove that for every connected $(P_6,C_6)$-free graph, $\gamma_c \le \gamma + 1$ holds, and there is a family of $(P_6,C_6)$-free graphs with arbitrarily large values of $\gamma$ attaining this bound.
Moreover, for every connected $(P_8,C_8)$-free graph, $\gamma_c / \gamma \le 2$, and there is a family of $(P_7,C_7)$-free graphs with arbitrarily large values of $\gamma$ attaining this bound.
In the class of $(P_9,C_9)$-free graphs, the general bound $\gamma_c / \gamma \le 3$ is asymptotically sharp.

\noindent \textbf{keywords:} domination, connected domination, forbidden induced subgraphs.

\end{abstract}

\section{Introduction}

A \textit{dominating set} of a graph $G$ is a vertex subset $X$ such that every vertex not in $X$ has a neighbor in $X$.
The minimum size of a dominating set of $G$ is called the \textit{domination number} of $G$ and is denoted by $\gamma(G)$.
A dominating set of size $\gamma(G)$ is called a \textit{minimum dominating set}.

Dominating sets have been intensively studied in the literature.
The main interest in dominating sets is due to their relevance on both theoretical and practical side.
Moreover, there are interesting variants of domination and many of them are well-studied.
A good introduction into the topic is given by Haynes, Hedetniemi and Slater \cite{haynes}.

A \textit{connected dominating set} of a graph $G$ is a dominating set $X$ whose induced subgraph, henceforth denoted $G[X]$, is connected.
The minimum size of such a set of a connected graph $G$, the \textit{connected domination number} of $G$, is denoted by $\gamma_c(G)$.
A connected dominating set of size $\gamma_c(G)$ is called a \textit{minimum connected dominating set}.
A connected dominating set such that every proper subset is not a connected dominating set is called \textit{minimal connected dominating set}.
Among the applications of connected dominating sets is the routing of messages in mobile ad-hoc networks. 
Blum, Ding, Thaeler and Cheng \cite{MANETS} explain the usefulness of connected dominating sets in this context.

A first impression of the relation of $\gamma_c$ and $\gamma$ is given by Duchet and Meyniel \cite{Basic_inequality}.

\begin{observation}[Duchet and Meyniel \cite{Basic_inequality}] \label{obs basic inequality}
For every connected graph it holds that $\gamma_c \le 3 \gamma - 2$.
\end{observation}

As an immediate consequence of Observation \ref{obs basic inequality},
\begin{equation} \label{eqn general poc}
 \gamma_c / \gamma < 3.
\end{equation}
Loosely speaking, the price of connectivity for minimum dominating sets, $ \gamma_c / \gamma $,  is strictly bounded by 3.

Let $P_k$ be the induced path on $k$ vertices and let $C_k$ be the induced cycle on $k$ vertices.
It is easy to see that
\begin{equation}\label{eqn paths and cycles}
\lim\limits_{k \rightarrow \infty} \gamma_c (P_k) / \gamma (P_k) = 3 = \lim\limits_{k \rightarrow \infty} \gamma_c (C_k) / \gamma (C_k).
\end{equation}
Hence, the upper bound (\ref{eqn general poc}) is asymptotically sharp in the class of paths and in the class of cycles.

The price of connectivity has been introduced by Cardinal and Levy~\cite{cardinal-levy,levy} for the vertex cover problem. They showed that it was bounded by $2/(1+\varepsilon) $ in graphs with average degree $\varepsilon n$, where $n$ is the number of vertices.
In a companion paper to the present paper, the price of connectivity for vertex cover is studied by Camby, Cardinal, Fiorini and Schaudt~\cite{CCFS}.
In a similar spirit, Schaudt~\cite{schaudt} studied the ratio between the connected domination number and the total domination number. 
Fulman~\cite{fulman} and Zverovich~\cite{zverovich} investigated the ratio between the independence number and the upper domination number. 
Many results in this area concern graph classes defined by forbidden induced subgraphs.
This line of research stems from the classical theory of perfect graphs, for which the clique number and the chromatic number are equal in every induced subgraph~\cite{golumbic}.

Motivated by (\ref{eqn paths and cycles}), in this paper we study the interdependence of $\gamma_c$ and $\gamma$ in graph classes defined by forbidden induced paths and cycles.
For this we use the following standard notation.
If $G$ and $H$ are two graphs, we say that $G$ is \textit{$H$-free} if $H$ does not appear as an induced subgraph of $G$.
Furthermore, if $G$ is $H_1$-free and $H_2$-free for some graphs $H_1$ and $H_2$, we say that $G$ is \textit{$(H_1,H_2)$-free}. 
Our starting point is the following result by Zverovich \cite{Connected-Dominant}.

\begin{theorem}[Zverovich \cite{Connected-Dominant}] \label{thm zverovich}
The following assertions are equivalent for every graph $G$.
\begin{enumerate}[(i)]
 \item For every connected induced subgraph of $G$ it holds that $\gamma_c = \gamma$.

 \item $G$ is $(P_5,C_5)$-free.
\end{enumerate}
\end{theorem}

We aim for similar bounds in the class of $(P_k,C_k)$-free connected graphs for $k \ge 6$.
The properties of connected dominating sets in $P_k$-free graphs have been studied before, e.g.~by Liu, Peng and Zhao \cite{LPZ} and later van 't Hof and Paulusma \cite{PP}.

We prove the following.
\begin{itemize}
	\item For connected $(P_6,C_6)$-free graphs, $\gamma_c \le \gamma + 1$, and this bound is attained by connected $(P_6,C_6)$-free graphs with arbitrarily large values of $\gamma$.

	\item For connected $(P_8,C_8)$-free graphs, $\gamma_c \le 2 \gamma$, and this bound is attained by connected $(P_7,C_7)$-free graphs with arbitrarily large values of $\gamma$. 
(In particular, the bound $\gamma_c \le 2 \gamma$ is best possible also in the class of connected $(P_7,C_7)$-free graphs.)

	\item For connected $(P_9,C_9)$-free graphs, the general bound $\gamma_c / \gamma < 3$ is asymptotically sharp.
\end{itemize}

Apart from the previous work, the research of this paper has an algorithmic motivation.
The proofs of our results are constructive in the sense that it is possible to draw polynomial time algorithms from them.
These algorithms can be used to build, given a dominating set of size $k$, a connected dominating set of size at most $f(k)$, for the suitable function $f$ provided by the respective theorem.
We do not explicitely give the algorithms in this paper, but leave it as a possible future application of our results.

\section{Our Results}

We will use the following lemma several times.
Note that this lemma is concerned with minimal connected dominating sets.

\begin{lemma} \label{lem short paths in CDS}
Let $G$ be a connected graph that is $(P_k,C_k)$-free for some $k \ge 4$ and let $X$ be a minimal connected dominating set of $G$.
Then $G[X]$ is $P_{k-2}$-free.
\end{lemma}

\begin{proof}
Let $G=(V,E)$ be a connected $(P_k,C_k)$-free graph and let $X$ be a minimal connected dominating set of $G$.
Suppose that there is an induced path $(v_1,v_2, \ldots , v_{k-2})$ on $k-2$ vertices in $G[X]$.
As $X$ is minimal, $X \setminus \{v_1\}$ is not a connected dominating set.
Hence, $X \setminus \{v_1\}$ is not a dominating set or $G[X \setminus \{v_1\}]$ is disconnected. 
In the first case, there is a vertex $v_1' \in V \setminus X$ whose only neighbor in $X$ is $v_1$.
In the second case, the vertices $v_2, \ldots , v_{k-2}$ are contained in a single connected component of $G[X \setminus \{v_1\}]$.
Thus, there is a neighbor of $v_1$ in $X$, say $v_1'$, that is not adjacent to any member of $\{v_2, \ldots , v_{k-2}\}$.
In both cases, there is a vertex $v_1' \notin \{v_1, v_2, \ldots , v_{k-2}\}$ whose only neighbor among $\{v_1, v_2, \ldots , v_{k-2}\}$ is $v_1$.
Similarly, there is a vertex $v_{k-2}' \notin \{v_1, v_2, \ldots , v_{k-2}\}$ whose only neighbor among $\{v_1, v_2, \ldots , v_{k-2}\}$ is $v_{k-2}$.
But then $G[\{v_1',v_1, v_2, \ldots , v_{k-2},v_{k-2}'\}] \cong P_k$ or $G[\{v_1',v_1, v_2, \ldots , v_{k-2},v_{k-2}'\}] \cong C_k$, depending on the adjacency of $v_1'$ and $v_{k-2}'$.
This is a contradiction to the choice of $G$.
\end{proof}

For a graph $G$ and $v \in V(G)$ we denote by $N_G[v]$ the closed neighborhood of $v$ in $G$.
Our first result establishes the upper bound $\gamma_c \le \gamma + 1$ in the class of connected $(P_6,C_6)$-free graphs.

\begin{theorem}\label{thm P6 C6 free}
For every connected $(P_6,C_6)$-free graph it holds that $\gamma_c \le \gamma + 1$.
\end{theorem}

\begin{proof}
Let $G = (V,E)$ be a connected $(P_6,C_6)$-free graph and let $D$ be a minimum dominating set of $G$.
Let $D_1, D_2, \ldots, D_k$ be the vertex sets of the connected components of $G[D]$.
Let $C \subseteq V$ be an inclusionwise minimal set such that $G[D \cup C]$ is connected, and let $X \subseteq D \cup C$ be a minimal connected dominating set.
By Lemma \ref{lem short paths in CDS}, $G[X]$ is $P_4$-free.

Let $I \subseteq \{1,2,\ldots,k\}$ be such that $i \in I$ if and only if $D_i \cap X = \emptyset$.
For every $i \in I$, pick $x_i \in X$ such that $x_i$ has a neighbor in $D_i$ (this is always possible, since $X$ is a dominating set).
The $x_i$ do not have to be distinct.
Let $S =  \bigcup_{i \notin I} (D_i \cap X) \cup \{x_j : j \in I\}$.

Assume first that $G[S]$ is connected.
Hence, $G[D \cup \{x_i : i \in I\}]$ is connected, and so $C = \{x_i : i \in I\}$.
Thus, $X = S$, which gives 
\[
\gamma_c \le |S| \le \sum_{i \notin I} |D_i \cap X| + |I| \le |D| = \gamma.
\]

So we may assume that $G[S]$ is not connected.
Among other authors, Seinsche~\cite{Seinsche} proved that every $P_4$-free graph with at least two vertices is either disconnected, or its complement is disconnected.
In particular, this applies to $G[X]$.
Since the complement of $G[S]$ is connected, but the complement of $G[X]$ is disconnected, there is a vertex $y \in X$ adjacent to every member of $S$.
Hence, $G[D \cup \{x_i : i \in I\} \cup \{y\}]$ is connected, and so $C = \{x_i : i \in I\} \cup \{y\}$.
Thus, $X = S \cup \{y\}$, which gives 
\[
\gamma_c \le |S| \le \sum_{i \notin I} |D_i \cap X| + |I| + 1 \le |D| +1 = \gamma +1.
\]
This completes the proof.
\end{proof}

To see that the bound given by Theorem \ref{thm P6 C6 free} is best possible, consider the following family of connected $(P_6,C_6)$-free graphs.
For each $k \in \mathbb{N}$, let $F_k$ be the graph obtained from a $k$ disjoint copies of $C_4$, by picking one vertex from every copy and identifying these picked vertices to a single vertex $x$, and afterwards attaching a path of length 2 to $x$ (see Fig.~\ref{fig P6 C6}).
It is easy to see that, for all $k$, $\gamma_c(F_k)=\gamma(F_k)+1$.
Moreover, the graph $F_k$ does not have a minimum connected dominating set that contains a minimum dominating set as subset.

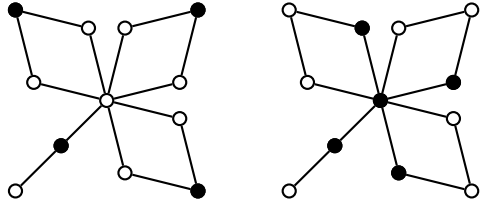
\begin{figure}[ht]
\begin{center}
\psset{unit=1.2cm}
\begin{pspicture}(0,0)(5,2)

\cnode(0,0){0.1cm}{a1}
\cnode(0.5,0.5){0.1cm}{a2}

\pscircle*(0.5,0.5){0.1cm}

\cnode(0,2){0.1cm}{b1}
\cnode(0.8,1.8){0.1cm}{b2}
\cnode(0.2,1.2){0.1cm}{b3}

\pscircle*(0,2){0.1cm}

\cnode(2,2){0.1cm}{c1}
\cnode(1.8,1.2){0.1cm}{c2}
\cnode(1.2,1.8){0.1cm}{c3}

\pscircle*(2,2){0.1cm}

\cnode(2,0){0.1cm}{d1}
\cnode(1.2,0.2){0.1cm}{d2}
\cnode(1.8,0.8){0.1cm}{d3}

\pscircle*(2,0){0.1cm}

\cnode(1,1){0.1cm}{x}

\ncarc[arcangle=0]{-}{a1}{a2}
\ncarc[arcangle=0]{-}{a2}{x}

\ncarc[arcangle=0]{-}{b1}{b2}
\ncarc[arcangle=0]{-}{b1}{b3}
\ncarc[arcangle=0]{-}{b2}{x}
\ncarc[arcangle=0]{-}{b3}{x}

\ncarc[arcangle=0]{-}{c1}{c2}
\ncarc[arcangle=0]{-}{c1}{c3}
\ncarc[arcangle=0]{-}{c2}{x}
\ncarc[arcangle=0]{-}{c3}{x}

\ncarc[arcangle=0]{-}{d1}{d2}
\ncarc[arcangle=0]{-}{d1}{d3}
\ncarc[arcangle=0]{-}{d2}{x}
\ncarc[arcangle=0]{-}{d3}{x}

\cnode(3,0){0.1cm}{a1}
\cnode(3.5,0.5){0.1cm}{a2}

\pscircle*(3.5,0.5){0.1cm}

\cnode(3,2){0.1cm}{b1}
\cnode(3.8,1.8){0.1cm}{b2}
\cnode(3.2,1.2){0.1cm}{b3}

\pscircle*(3.8,1.8){0.1cm}

\cnode(5,2){0.1cm}{c1}
\cnode(4.8,1.2){0.1cm}{c2}
\cnode(4.2,1.8){0.1cm}{c3}

\pscircle*(4.8,1.2){0.1cm}

\cnode(5,0){0.1cm}{d1}
\cnode(4.2,0.2){0.1cm}{d2}
\cnode(4.8,0.8){0.1cm}{d3}

\pscircle*(4.2,0.2){0.1cm}

\cnode(4,1){0.1cm}{x}

\pscircle*(4,1){0.1cm}

\ncarc[arcangle=0]{-}{a1}{a2}
\ncarc[arcangle=0]{-}{a2}{x}

\ncarc[arcangle=0]{-}{b1}{b2}
\ncarc[arcangle=0]{-}{b1}{b3}
\ncarc[arcangle=0]{-}{b2}{x}
\ncarc[arcangle=0]{-}{b3}{x}

\ncarc[arcangle=0]{-}{c1}{c2}
\ncarc[arcangle=0]{-}{c1}{c3}
\ncarc[arcangle=0]{-}{c2}{x}
\ncarc[arcangle=0]{-}{c3}{x}

\ncarc[arcangle=0]{-}{d1}{d2}
\ncarc[arcangle=0]{-}{d1}{d3}
\ncarc[arcangle=0]{-}{d2}{x}
\ncarc[arcangle=0]{-}{d3}{x}

\end{pspicture}
\end{center}
\caption{The black vertices indicate a minimum dominating set (resp.~a minimum connected dominating set) of $F_3$.}
\label{fig P6 C6}
\end{figure}

\begin{theorem}\label{thm P8C8 free}
For every connected $(P_8,C_8)$-free graph it holds that $\gamma_c / \gamma \le 2$.
\end{theorem}

\begin{proof}
Let $G = (V,E)$ be a connected $(P_8,C_8)$-free graph and let $D$ be a minimum dominating set of $G$.
Let $D_1, D_2, \ldots, D_k$ be the vertex sets of the connected components of $G[D]$.
It is clear that we can assume that $G[D]$ has at least two components, i.e., $k \ge 2$.
Let $C \subseteq V$ be an inclusionwise minimal set such that $G[D \cup C]$ is connected, and let $X \subseteq D \cup C$ be a minimal connected dominating set.
By Lemma \ref{lem short paths in CDS}, $G[X]$ is $P_6$-free.

Let us first assume that $\gamma \le 3$.
If $k = 2$, $\gamma_c \le \gamma + 2$ and so clearly $\gamma_c / \gamma \le 2$ holds.
So we may assume that $k=3$, i.e., $D$ consists of three isolated vertices, say $x$, $y$, and $z$.
Since $D$ is a dominating set, $|C| \le 4$, and so $|D \cup C| \le 7$.
Suppose that $|X| = 7$, i.e., $X = D \cup C$.
Thus, the distance in $G[X]$ between any two vertices among $x$, $y$, and $z$ is exactly 2.
Let $(x,u,v,y)$ be a shortest path in $G[X]$ between $x$ and $y$, and let $(x,u',v',z)$ be a shortest path in $G[X]$ between $x$ and $z$.
Thus $C = \{u,v,u',v'\}$.
Since $G[X]$ is $P_6$-free, the distance between $v$ and $z$ in $G[X]$ is at most $4$. Thus $v$ or $u$ must be adjacent to a member of $\{u',v',z\}$.
Note that if there is only one adjacency and if this adjacency is between $u$ and $u'$, $G[\{y,v,u,u',v',z\}]$ is an induced $P_6$, a contradiction. In the other cases, $X$ is not minimal, again a contradiction.

Now let us assume that $\gamma \ge 4$.

Let $I \subseteq \{1,2,\ldots,k\}$ be such that $i \in I$ if and only if $D_i \cap X = \emptyset$.
For every $i \in I$, pick $x_i \in X$ such that $x_i$ has a neighbor in $D_i$ (this is always possible, since $X$ is a dominating set).
The $x_i$ do not have to be distinct.
Let $S =  \bigcup_{i \notin I} (D_i \cap X) \cup \{x_j : j \in I\}$.

It is shown by van t' Hof and Paulusma \cite{PP} that every connected $P_6$-free graph has a connected dominating set $Z$ for which the following holds: either $G[Z] \cong C_6$ or $G[Z]$ contains a complete bipartite graph as spanning subgraph.
Let $Y$ be such a connected dominating set for $G[X]$.

Assume first that $G[Y] \cong C_6$.
Let $u_1, u_2 , \ldots , u_6$ be a consecutive ordering of the vertices of the $C_6$.
Suppose that $Y' = \{ u_1 , u_2 , u_3 , u_4 \}$ is not a dominating set of $G[X]$.
Thus there is a vertex $z \in X$ with $N[z] \cap Y' = \emptyset$.
Without loss of generality, $z$ is adjacent to $u_5$.
But then $G[Y' \cap \{u_5 , z\}] \cong P_6$, a contradiction to the fact that $G[X]$ is $P_6$-free.
Thus $Y'$ is a connected dominating set of $G[X]$.

Since $G[D \cup Y' \cup \{x_j : j \in I\}]$ is connected, $C \subseteq Y' \cup \{x_j : j \in I\}$.
Thus, $X \subseteq S \cup Y'$, which gives 
\[
\gamma_c \le |S| + |Y'| \le \sum_{i \notin I} |D_i \cap X| + |I| + 4 \le |D| + 4 \le 2 \gamma.
\]

So we may assume that $G[Y]$ contains a complete bipartite graph as spanning subgraph.
Let $(A,B)$ be a bipartition of this complete bipartite graph.
For each $1 \le i \le k$, pick $y_i \in Y$ with the following property.
If $D_i \cap X \neq \emptyset$, $N_{G[X]}[y_i] \cap D_i \neq \emptyset$, and if $D_i \cap X = \emptyset$, $y_i \in N_{G[X]}[x_i]$.
These $y_i$ exist since $Y$ is a dominating set of $G[X]$.
We can assume that $A \cap \{y_i : 1 \le i \le k\} \neq \emptyset$.

If $B \cap \{y_i : 1 \le i \le k\} \neq \emptyset$, $G[D \cup \{x_i : i \in I\} \cup \{y_j : 1 \le j \le k\}]$ is connected.
So, $C = \{x_i : i \in I\} \cup \{y_j : 1 \le j \le k\}$ and thus $X = S \cup \{y_j : 1 \le j \le k\}$, which gives 
\begin{equation}
\label{eqn:P8C8-free}
\gamma_c \le |S| + k \le \sum_{i \notin I} |D_i \cap X| + |I| + k \le |D| + k \le 2 \gamma.
\end{equation}

So we may assume that $B \cap \{y_i : 1 \le i \le k\} = \emptyset$.
Pick any $z \in B$.
Since $D$ is a dominating set, there is an index $1 \le l \le k$ such that $N[z] \cap D_l \neq \emptyset$.
So, $A \cap \{y_i : 1 \le i \le k, i \neq l\} \neq \emptyset$.
Hence, $G[D \cup \{x_i : i \in I\} \cup \{y_j : 1 \le j \le k, j \neq l\} \cup \{z\}]$ is connected.
So, $C = \{x_i : i \in I\} \cup \{y_j : 1 \le j \le k, j \neq l\} \cup \{z\}$ and thus $X = S \cup \{y_j : 1 \le j \le k, j \neq l\} \cup \{z\}$, which gives (\ref{eqn:P8C8-free}).
This completes the proof.
\end{proof}

The bound provided by Theorem \ref{thm P8C8 free} is attained by an infinite number of connected $(P_7,C_7)$-free graphs, given by the following construction.
For every $k \in \mathbb{N}$, let $H_k$ be the graph defined as follows (cf.~Fig.~\ref{fig P7}).
Start with $k$ paths $P^1$, $P^2$, \ldots , $P^k$ on three vertices each.
For every $1 \le i \le k$, choose an end-vertex $v_i$ of $P^i$.
Let $H_k$ be the graph obtained from the disjoint union of all $P^i$, $1 \le i \le k$, by adding all possible edges between the vertices $v_i$, $1 \le i \le k$.
So, $H_k[\{v_i : 1 \le i \le k\}]$ is a complete graph.
It is easily seen that, for all $k \in \mathbb{N}$, $\gamma_c(H_k) / \gamma(H_k) = 2$.

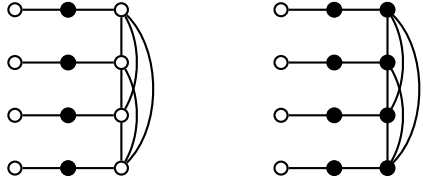
\begin{figure}[ht]
\begin{center}
\psset{unit=0.7cm}
\begin{pspicture}(0,0)(8,3)

\cnode(0,0){0.1cm}{a2}
\cnode(1,0){0.1cm}{a3}
\cnode(2,0){0.1cm}{a4}

\cnode(0,1){0.1cm}{b2}
\cnode(1,1){0.1cm}{b3}
\cnode(2,1){0.1cm}{b4}

\cnode(0,2){0.1cm}{c2}
\cnode(1,2){0.1cm}{c3}
\cnode(2,2){0.1cm}{c4}

\cnode(0,3){0.1cm}{d2}
\cnode(1,3){0.1cm}{d3}
\cnode(2,3){0.1cm}{d4}

\ncarc[arcangle=0]{-}{a2}{a3}
\ncarc[arcangle=0]{-}{a3}{a4}

\ncarc[arcangle=0]{-}{b2}{b3}
\ncarc[arcangle=0]{-}{b3}{b4}

\ncarc[arcangle=0]{-}{c2}{c3}
\ncarc[arcangle=0]{-}{c3}{c4}

\ncarc[arcangle=0]{-}{d2}{d3}
\ncarc[arcangle=0]{-}{d3}{d4}

\ncarc[arcangle=0]{-}{a4}{b4}
\ncarc[arcangle=0]{-}{b4}{c4}
\ncarc[arcangle=0]{-}{c4}{d4}

\ncarc[arcangle=-30]{-}{a4}{c4}
\ncarc[arcangle=-30]{-}{b4}{d4}
\ncarc[arcangle=-45]{-}{a4}{d4}

\pscircle*(1,0){0.15}
\pscircle*(1,1){0.15}
\pscircle*(1,2){0.15}
\pscircle*(1,3){0.15}


\cnode(5,0){0.1cm}{aa2}
\cnode(6,0){0.1cm}{aa3}
\cnode(7,0){0.1cm}{aa4}

\cnode(5,1){0.1cm}{bb2}
\cnode(6,1){0.1cm}{bb3}
\cnode(7,1){0.1cm}{bb4}

\cnode(5,2){0.1cm}{cc2}
\cnode(6,2){0.1cm}{cc3}
\cnode(7,2){0.1cm}{cc4}

\cnode(5,3){0.1cm}{dd2}
\cnode(6,3){0.1cm}{dd3}
\cnode(7,3){0.1cm}{dd4}

\ncarc[arcangle=0]{-}{aa2}{aa3}
\ncarc[arcangle=0]{-}{aa3}{aa4}

\ncarc[arcangle=0]{-}{bb2}{bb3}
\ncarc[arcangle=0]{-}{bb3}{bb4}

\ncarc[arcangle=0]{-}{cc2}{cc3}
\ncarc[arcangle=0]{-}{cc3}{cc4}

\ncarc[arcangle=0]{-}{dd2}{dd3}
\ncarc[arcangle=0]{-}{dd3}{dd4}

\ncarc[arcangle=0]{-}{aa4}{bb4}
\ncarc[arcangle=0]{-}{bb4}{cc4}
\ncarc[arcangle=0]{-}{cc4}{dd4}

\ncarc[arcangle=-30]{-}{aa4}{cc4}
\ncarc[arcangle=-30]{-}{bb4}{dd4}
\ncarc[arcangle=-45]{-}{aa4}{dd4}

\pscircle*(6,0){0.15}
\pscircle*(7,0){0.15}

\pscircle*(6,1){0.15}
\pscircle*(7,1){0.15}

\pscircle*(6,2){0.15}
\pscircle*(7,2){0.15}

\pscircle*(6,3){0.15}
\pscircle*(7,3){0.15}

\end{pspicture}
\end{center}
\caption{The black vertices indicate a minimum dominating set (resp.~a minimum connected dominating set) of $H_4$.}
\label{fig P7}
\end{figure}

A similar construction shows that (\ref{eqn general poc}) is asymptotically sharp in the class of connected $(P_9,C_9)$-free graphs, in the sense that there is a family $\{G_k : k \in \mathbb{N}\}$ of $(P_9,C_9)$-free graphs such that $\lim_{k \rightarrow \infty} \gamma_c (G_k) / \gamma (G_k) = 3$.
For every $k \in \mathbb{N}$ let $G_k$ be the graph obtained by attaching a pendant vertex to every pendant vertex of $H_k$ (cf.~Fig.~\ref{fig P9 C9}).
It is easy to check that for every $k \ge 2$, $\gamma (G_k) = k+1$ and $\gamma_c (G_k) = 3k$.
Furthermore, $G_k$ is $(P_9,C_9)$-free.

\begin{figure}[ht]
\begin{center}
\psset{unit=0.7cm}
\begin{pspicture}(0,0)(10,3)

\cnode(0,0){0.1cm}{a1}
\cnode(1,0){0.1cm}{a2}
\cnode(2,0){0.1cm}{a3}
\cnode(3,0){0.1cm}{a4}

\cnode(0,1){0.1cm}{b1}
\cnode(1,1){0.1cm}{b2}
\cnode(2,1){0.1cm}{b3}
\cnode(3,1){0.1cm}{b4}

\cnode(0,2){0.1cm}{c1}
\cnode(1,2){0.1cm}{c2}
\cnode(2,2){0.1cm}{c3}
\cnode(3,2){0.1cm}{c4}

\cnode(0,3){0.1cm}{d1}
\cnode(1,3){0.1cm}{d2}
\cnode(2,3){0.1cm}{d3}
\cnode(3,3){0.1cm}{d4}

\ncarc[arcangle=0]{-}{a1}{a2}
\ncarc[arcangle=0]{-}{a2}{a3}
\ncarc[arcangle=0]{-}{a3}{a4}

\ncarc[arcangle=0]{-}{b1}{b2}
\ncarc[arcangle=0]{-}{b2}{b3}
\ncarc[arcangle=0]{-}{b3}{b4}

\ncarc[arcangle=0]{-}{c1}{c2}
\ncarc[arcangle=0]{-}{c2}{c3}
\ncarc[arcangle=0]{-}{c3}{c4}

\ncarc[arcangle=0]{-}{d1}{d2}
\ncarc[arcangle=0]{-}{d2}{d3}
\ncarc[arcangle=0]{-}{d3}{d4}

\ncarc[arcangle=0]{-}{a4}{b4}
\ncarc[arcangle=0]{-}{b4}{c4}
\ncarc[arcangle=0]{-}{c4}{d4}

\ncarc[arcangle=-30]{-}{a4}{c4}
\ncarc[arcangle=-30]{-}{b4}{d4}
\ncarc[arcangle=-45]{-}{a4}{d4}

\pscircle*(1,0){0.15}
\pscircle*(1,1){0.15}
\pscircle*(1,2){0.15}
\pscircle*(1,3){0.15}
\pscircle*(3,0){0.15}


\cnode(6,0){0.1cm}{aa1}
\cnode(7,0){0.1cm}{aa2}
\cnode(8,0){0.1cm}{aa3}
\cnode(9,0){0.1cm}{aa4}

\cnode(6,1){0.1cm}{bb1}
\cnode(7,1){0.1cm}{bb2}
\cnode(8,1){0.1cm}{bb3}
\cnode(9,1){0.1cm}{bb4}

\cnode(6,2){0.1cm}{cc1}
\cnode(7,2){0.1cm}{cc2}
\cnode(8,2){0.1cm}{cc3}
\cnode(9,2){0.1cm}{cc4}

\cnode(6,3){0.1cm}{dd1}
\cnode(7,3){0.1cm}{dd2}
\cnode(8,3){0.1cm}{dd3}
\cnode(9,3){0.1cm}{dd4}

\ncarc[arcangle=0]{-}{aa1}{aa2}
\ncarc[arcangle=0]{-}{aa2}{aa3}
\ncarc[arcangle=0]{-}{aa3}{aa4}

\ncarc[arcangle=0]{-}{bb1}{bb2}
\ncarc[arcangle=0]{-}{bb2}{bb3}
\ncarc[arcangle=0]{-}{bb3}{bb4}

\ncarc[arcangle=0]{-}{cc1}{cc2}
\ncarc[arcangle=0]{-}{cc2}{cc3}
\ncarc[arcangle=0]{-}{cc3}{cc4}

\ncarc[arcangle=0]{-}{dd1}{dd2}
\ncarc[arcangle=0]{-}{dd2}{dd3}
\ncarc[arcangle=0]{-}{dd3}{dd4}

\ncarc[arcangle=0]{-}{aa4}{bb4}
\ncarc[arcangle=0]{-}{bb4}{cc4}
\ncarc[arcangle=0]{-}{cc4}{dd4}

\ncarc[arcangle=-30]{-}{aa4}{cc4}
\ncarc[arcangle=-30]{-}{bb4}{dd4}
\ncarc[arcangle=-45]{-}{aa4}{dd4}

\pscircle*(7,0){0.15}
\pscircle*(8,0){0.15}
\pscircle*(9,0){0.15}

\pscircle*(7,1){0.15}
\pscircle*(8,1){0.15}
\pscircle*(9,1){0.15}

\pscircle*(7,2){0.15}
\pscircle*(8,2){0.15}
\pscircle*(9,2){0.15}

\pscircle*(7,3){0.15}
\pscircle*(8,3){0.15}
\pscircle*(9,3){0.15}

\end{pspicture}
\end{center}
\caption{The black vertices indicate a minimum dominating set (resp.~a minimum connected dominating set) of $G_4$.}
\label{fig P9 C9}
\end{figure}
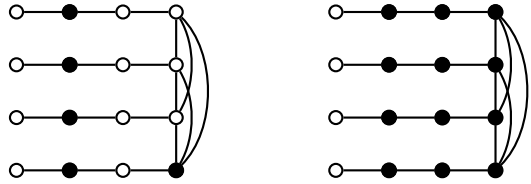

\section{A Conjecture}

We close the paper with a conjecture that came up during our research.
As Theorem~\ref{thm P8C8 free} shows, $\gamma_c \le 2 \gamma$ holds in every connected $(P_8,C_8)$-free graph.
However, $\gamma_c(P_8)/\gamma(P_8) = 2 = \gamma_c(C_8)/\gamma(C_8)$, i.e., both $P_8$ and $C_8$ do not violate the bound given by Theorem~\ref{thm P8C8 free}. 

\begin{conjecture}\label{conj P9 C9 H}
For every connected $(P_9,C_9,H)$-free graph, $\gamma_c \le 2 \gamma$ (see Fig.~\ref{fig H} for $H$).
\end{conjecture}

Note that $P_9$, $C_9$, and $H$ violate $\gamma_c \le 2 \gamma$.
Hence, if true, Conjecture~\ref{conj P9 C9 H} would give a characterization of the largest graph class that is closed under connected induced subgraphs where $\gamma_c \le 2 \gamma$ holds.

\begin{figure}[ht]
\begin{center}
\begin{pspicture}(0,0)(4.8,1)

\cnode(3,0){0.1cm}{a1}
\cnode(2,0){0.1cm}{a2}
\cnode(1,0){0.1cm}{a3}
\cnode(0,0){0.1cm}{a4}

\cnode(3,1){0.1cm}{b1}
\cnode(2,1){0.1cm}{b2}
\cnode(1,1){0.1cm}{b3}
\cnode(0,1){0.1cm}{b4}

\cnode(3.8,0.5){0.1cm}{c}
\cnode(4.8,0.5){0.1cm}{d}

\ncarc[arcangle=0]{-}{a1}{a2}
\ncarc[arcangle=0]{-}{a2}{a3}
\ncarc[arcangle=0]{-}{a3}{a4}

\ncarc[arcangle=0]{-}{b1}{b2}
\ncarc[arcangle=0]{-}{b2}{b3}
\ncarc[arcangle=0]{-}{b3}{b4}

\ncarc[arcangle=0]{-}{a1}{b1}
\ncarc[arcangle=0]{-}{a1}{c}
\ncarc[arcangle=0]{-}{c}{b1}
\ncarc[arcangle=0]{-}{c}{d}

\end{pspicture}
\caption{The graph $H$ from Conjecture~\ref{conj P9 C9 H}.}
\label{fig H}
\end{center}
\end{figure}
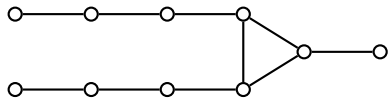

\end{document}